\documentclass[preprint]{article}

\PassOptionsToPackage{numbers, sort&compress}{natbib}
\usepackage{neurips_2026}
\usepackage[utf8]{inputenc}
\usepackage[T1]{fontenc}
\usepackage[colorlinks=true, allcolors=blue]{hyperref}
\usepackage{graphicx}
\usepackage{url}
\usepackage{booktabs}
\usepackage{amsfonts}
\usepackage{microtype}
\usepackage{xcolor}
\usepackage{amsmath}
\usepackage{amsthm}
\usepackage{enumitem}
\setlist[itemize]{label=$*$, leftmargin=*}
\usepackage{nicefrac}
\usepackage{pifont}

\newcommand{\pp}{\textsc{\small {PsychoPass}}~}

\newcommand{\lex}{\ensuremath{\phi_{L}}}
\newcommand{\sem}{\ensuremath{\phi_{S}}}
\newcommand{\Traj}{\ensuremath{\mathcal{T}}}
\newcommand{\TrajU}{\ensuremath{\mathcal{T}_{U}}}
\newcommand{\TrajA}{\ensuremath{\mathcal{T}_{A}}}
\newcommand{\TrajL}{\ensuremath{\mathcal{T}^{(-k)}}}
\newcommand{\TrajF}[1]{\ensuremath{\mathcal{T}^{(:#1)}}}

\newcommand{\Texec}{\ensuremath{T_{\operatorname{exec}}}}
\newcommand{\Tmax}{\ensuremath{T_{\max}}}
\newcommand{\Lpath}{\ensuremath{L}}
\newcommand{\Disp}{\ensuremath{D}}
\newcommand{\directness}{\ensuremath{\eta}}
\newcommand{\meanSpeed}{\ensuremath{\bar{s}}}
\newcommand{\stdSpeed}{\ensuremath{\sigma_s}}
\newcommand{\Vel}{\ensuremath{V}}
\newcommand{\Circ}{\ensuremath{\zeta}}
\newcommand{\OTpos}{\ensuremath{\mathrm{OT}^{+}}}
\newcommand{\OTneg}{\ensuremath{\mathrm{OT}^{-}}}
\newcommand{\TRA}{\ensuremath{\mathrm{TRA}}}
\newcommand{\StretchHigh}{\ensuremath{\mathrm{Str^{\uparrow}}}}
\newcommand{\StretchDec}{\ensuremath{\mathrm{Str^{\downarrow}}}}
\newcommand{\HighFluc}{\ensuremath{\mathrm{HF}}}

\newcommand{\cmark}{\ding{51}}
\newcommand{\xmark}{\ding{55}}

\theoremstyle{plain}
\newtheorem{definition}{Definition}
\newtheorem{assumption}{Assumption}

\newtheorem{theorem}{Theorem}
\newtheorem{proposition}[theorem]{Proposition}

\title{\textsc{PsychoPass}: Geometric Profiling of Multi-Turn Adversarial LLM Conversations}

\author{
    Muberra Ozmen \\
    Coveo \\
    Montreal, QC, Canada \\
    \texttt{mozmen@coveo.com} \\
    \And
    Subhabrata Majumdar \\
    Indian Institute of Management Bangalore \\
    Bengaluru, India \\
    \texttt{smajumdar@iimb.ac.in} \\
}

\begin{document}

\maketitle

\begin{abstract}
    Multi-turn jailbreak attacks on large language models (LLMs) reveal a mismatch in current guardrails: they operate on individual turns, while attacks unfold as trajectories across conversations. 
    We propose a shift from content to dynamics, modeling conversations as paths in representation space and asking whether adversarial intent is encoded early in their geometry. 
    We introduce \textsc{\small PsychoPass}, a framework that extracts geometric features from conversation trajectories in embedding space to predict a potential attack \emph{before} harmful content is produced. 
    These features achieve near-perfect performance in na\"ive classifiers, which is largely explained by the inclusion of number of turns as a feature. 
    After removing this confound, a smaller but consistent geometric signal remains, with classification performance that does not depend meaningfully on encoder choice. 
    Crucially, this signal appears early in the conversation: attack outcomes remain above chance from short prefixes alone, more reliably than baseline guardrails.  
    A supporting theoretical analysis explains these findings via a decomposition of length and shape, a detection bound based on prefix length, and encoder invariance. 
    Together, these results show that adversarial conversations leave an early, representation-robust geometric fingerprint suitable for online monitoring \footnote{Link to repository: \hyperlink{https://github.com/muberraozmen/psycho-pass}{https://github.com/muberraozmen/psycho-pass}}.
\end{abstract}

\section{Introduction}
Large Language Models (LLMs) deployed in conversational interfaces are routinely targeted by multi-turn jailbreaks that escalate, backtrack, and reframe a forbidden request across several turns until the target system complies~\citep{russinovich2025greatwritearticlethat, mehrotra2024tap}.
Most runtime defenses aimed at preventing such attacks operate turn-by-turn: each user message is scored independently for policy violation, and each model response is filtered after generation. 
This design is fundamentally reactive, in that it can flag harmful \emph{content} once it appears but cannot flag the \emph{process} by which an attacker steers a conversation toward harm before that content is produced.

This paper asks whether the geometry of the trajectory of a conversation in a representational space can serve as an early warning signal for adversarial intent. 
Our premise is that if a multi-turn attack must systematically move a conversation toward a forbidden region, that directed movement should leave a geometric fingerprint, such as through curvature, drift, circularity, or temporal asymmetry of the conversation trajectory that distinguishes it from benign interaction before any single turn looks dangerous on its own (Figure~\ref{fig:pipeline}).

\begin{figure}[h]\label{fig:pipeline}
    \centering
    \includegraphics[width=\linewidth]{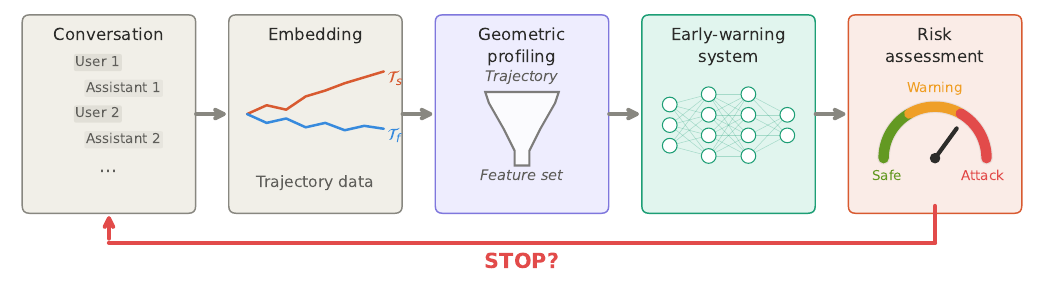}
    \caption{
        The \pp pipeline. 
        We distill down the difference between trajectories of successful ($\mathcal{T}_s$) and failed ($\mathcal{T}_f$) attacks in the embedding space using geometric features, which are used to train an early-warning system to provide online risk assessment to detect attacks in progress.
    }
\end{figure}

To test our proposition, we embed multi-turn adversarial conversations via a sparse lexical encoder and a dense semantic encoder, extract trajectory-level shape statistics and temporal features, and run a sequence of four experiments to progressively isolate the geometric signal from confounds. 
The first experiment shows that na\"ive trajectory classification achieves near-perfect AUROC, but this is almost entirely a length artifact: failed attacks exhaust their turn budget while successful ones terminate early, and any feature that scales with conversation length inherits this free signal. 
The second experiment removes the length confound by equalizing all trajectories to six turns. 
AUROC drops but remains reliably above chance, carried by stretch-decreasing dynamics, outlier timing, and low-frequency fluctuation; the choice of encoder has no statistically meaningful effect on classification. 
The third experiment asks whether this residual signal is available early enough to matter, and finds that it is: attack outcome stays above chance from prefixes as short as two turns, while a content-based safety classifier applied to the same prefixes collapses to chance.
We complement the empirical study with a theoretical analysis. 
A decomposition of discriminability into length and shape components quantitatively predicts the performance drop between the first and second experiments, a finite-sample detection bound explains why short prefixes suffice in the third experiment, and a rank-preservation result explains the encoder invariance in the second experiment.

\paragraph{Related work.}
How LLMs behave in the wild over long conversations spanning multiple turns is only beginning to be understood. 
Recently, \citet{gooding2025interactiondynamicsrewardsignal} showed that geometric features of conversation trajectories can serve as reward signals for alignment. 
Other studies have found that conversational history geometrically constrains future responses \cite{simhi2026oldhabitsdiehard}, model performance degrades as conversations grow longer \cite{laban2025llmslostmultiturnconversation}, and non-reasoning LLMs show improved performance when an input prompt is repeated inside the same conversation \cite{leviathan2025promptrepetitionimprovesnonreasoning}. 
These findings suggest that there are structural nuances in multi-turn interactions that single-turn analysis misses entirely.

Research on adversarial attacks on LLMs has followed a similar single-turn-to-multi-turn trajectory. 
Early jailbreak and prompt injection attacks relied on pre-crafted, static prefixes or suffixes \citep{zou2023universaltransferableadversarialattacks,liu2024autodangeneratingstealthyjailbreak}.
As guardrails and safety training became adept at stopping such attacks, researchers discovered that new strategies like Crescendo~\citep{russinovich2025greatwritearticlethat} and TAP~\citep{mehrotra2024tap} can spread and escalate an attack across turns to effectively circumvent perimeter defenses and make an LLM achieve a harmful objective. 
Several open-source libraries implement single- and multi-turn attacks. 
Two such well-known tools are garak~\cite{derczynski2024garakframeworksecurityprobing} and PyRIT~\citep{munoz2024pyritframeworksecurityrisk}, which standardize multi-turn attack pipelines using harmful objectives from safety benchmarks like AdvBench and HarmBench \citep{mazeika2024harmbenchstandardizedevaluationframework}.

Runtime defense methods to prevent attacks focus on filtering undesired inputs and outputs. 
Past research directions on such methods include encoder-only classifiers \cite{prompt-guard-86m,vijil-mbert}, moderation LLMs \cite{llama-guard-4-12b,granite-guardian-41}, embedding-based classifiers~\cite{ayub2024embeddingbasedclassifiersdetectprompt}, and inference-time methods~\cite{fath,spotlighting}. 
While tools like Nova rules \cite{Nova} can combine them to operationalize a guardrail that achieves acceptable latency-performance tradeoff in deployed AI systems, these methods often generalize poorly to unseen data. 
For example, \citet{hackett2025bypassingllmguardrailsempirical} showed that the best open-source prompt injection classifiers remain vulnerable to evasion techniques that perturb a malicious input prompt to obfuscate its content.

Existing runtime defenses work by filtering harmful content, not by discouraging malicious intent. 
This leaves the LLM underneath vulnerable to an adaptable attacker---often another LLM---to try alternative ways of achieving the same harmful goal if an initial attack is foiled. 
\pp aims to mitigate this risk, by using information inherent in the shape of an adversarial trajectory to enable real-time intervention of attacks unfolding over multiple turns.

\section{Methodology}\label{sec:methodology}
The \pp pipeline follows three stages: \emph{Attack Generation}, which produces adversarial multi-turn conversations against target LLMs; \emph{Trajectory Construction}, which embeds each conversation as a path followed in $d$-dimension space; and \emph{Geometric Profiling}, which extracts statistical features from those paths. 

\paragraph{Attack generation}
We use PyRIT~\citep{munoz2024pyritframeworksecurityrisk} to orchestrate Crescendo~\citep{russinovich2025greatwritearticlethat} attacks, which utilizes a multi-turn escalation strategy with dynamic backtracking, against four target LLMs (see Appendix~\ref{app:crescendo} for details). 

Each attack is initialized with a seed objective drawn from one of the benchmarks bundled with PyRIT: AdvBench~\citep{zou2023universaltransferableadversarialattacks}, HarmBench~\citep{mazeika2024harmbenchstandardizedevaluationframework}, and AIRT~\citep{munoz2024pyritframeworksecurityrisk}. 
In total we draw seed from these three families, sampling two attacks per objective. 
Each attack is limited to a maximum of $\Tmax=8$ turns and $2$ backtracks, where one turn is defined as one user message followed by one assistant message.

Three LLMs are involved in the attack generation process. 
The \emph{adversarial bot} (\texttt{llama-3.1-8b-instruct}, temperature $t = 1.0$) acts as a proxy user that iteratively refines prompts to satisfy the seed objective. 
The \emph{objective bot} is the system under evaluation; we experiment with four target LLMs (\texttt{llama-3.1-8b-instruct}, \texttt{gpt-oss-120b}, \texttt{mistral-small-3.2-24b-instruct}, \texttt{qwen-2.5-7b-instruct}), all at $t=1.0$. 
The \emph{scoring bot} (\texttt{gpt-oss-120b}, $t=0.1$) is an automated judge that declares an attack a success when its score exceeds a threshold of $0.8$. 
All bots are configured with a top-$p$ value of $1.0$ and are restricted to $1{,}024$ completion tokens per turn.

\paragraph{Trajectory construction}
A conversation $C = \{m_i\}_{i=1}^{N}$ of $N = 2 \times \Texec$ messages (steps), where $\Texec$ denotes number of executed turns (with alternating user and assistant messages), is projected into a vector space via an encoder $\phi:\mathcal{M}\to\mathbb{R}^d$, yielding the \emph{conversation trajectory} $\Traj = \{v_i\}_{i=1}^{N}$ with $v_i = \phi(m_i)$. 
We use two encoders:
\begin{itemize}[leftmargin=*,nolistsep]
    \item \emph{Lexical} ($\lex$): TF-IDF with up to $d = 4{,}096$ features, English
    stop-word removal, and a max-document-frequency of $0.95$.
    \item \emph{Semantic} ($\sem$): the dense embedding model
    \texttt{qwen3-embedding-8b} with a context window of $32{,}000$ tokens which yields a dimensionality of $d = 4{,}096$.
\end{itemize}
To isolate role-specific signals and to support the predictive-horizon analysis, we additionally consider the \emph{user trajectory} $\TrajU$ (messages of adversarial bot), the \emph{assistant trajectory} $\TrajA$ (messages of objective bot), and the \emph{trimmed trajectory} obtained either by keeping the first $k$ turns $\TrajF{k} = \Traj_{1:2k}$, or by removing the last $k$ turns $\TrajL = \Traj_{1:(N-2k)}$.

\paragraph{Geometric profiling}
Seven $L_2$-norm based statistics (Table~\ref{tab:features} top half) are used to summarize the geometry of trajectory $\Traj$.
\emph{Path length} ($\Lpath$) measures total conversational effort; \emph{Displacement} ($\Disp$) the net drift; \emph{directness} ($\directness$) whether progression was goal-oriented or meandering; the \emph{Step pace} pair ($(\meanSpeed,\stdSpeed)$) the average speed and its variability; \emph{Velocity} ($\Vel$) the directional progress per step; \emph{Circularity} ($\Circ$) whether the trajectory loops back through previously visited regions of the embedding space. 

Catch22~\citep{lubba2019catch22canonicaltimeseriescharacteristics} is a well-known method for extracting temporal features from time-series data. 
We adapt this method to capture higher-order temporal structure in our attack trajectories. 
We apply a targeted subset of the Catch22 feature set to the trajectory $\Traj$ (Table~\ref{tab:features} bottom half). 
Since the embedding dimension $d$ is large, computation is restricted to the top-$K{=}100$ dimensions with highest variance.
Each feature is computed independently per dimension and then aggregated across dimensions by mean, maximum, minimum and standard deviation. 
\emph{Time-reversal asymmetry} ($\TRA$) is sensitive to the temporal arrow of a series; a nonzero value indicates statistical irreversibility. 
\emph{Outlier timing} ($\OTpos$, $\OTneg$) measures the median normalized index of above- and below-mean outliers across an increasing sequence of thresholds; values near zero or one mean outliers concentrate early or late. 
\emph{Stretch-high} ($\StretchHigh$) is the longest run of consecutive above-mean values and captures the largest sustained excursion into a particular region of the embedding space. 
\emph{Stretch-decreasing} ($\StretchDec$) is the longest run of non-positive first differences, i.e., the most sustained monotonic withdrawal in a latent dimension. 
\emph{High fluctuation} ($\HighFluc$) is the trace of the covariance of a $3{\times}3$ transition matrix on the autocorrelation-down-sampled series, with larger values indicating more persistent low-frequency fluctuation in the bin sequence.

\begin{table}[t]\label{tab:features}
    \centering
    \caption{
        Geometric features extracted from a trajectory $\Traj = \{v_i\}_{i=1}^{N}\subset\mathbb{R}^d$. 
        $\Delta d_i = \|v_{i+1} - v_i\|_2$ denotes the $i$-th step length, $\mu = N^{-1} \sum_i v_i$ the centroid, $R_i = \|v_i - \mu\|_2$ the radius from $\mu$, and $\epsilon > 0$ a small constant for numerical stability.
        $x = (x_1, \dots, x_N)$ denotes the $z$-normalized values of a single coordinate of $\Traj$, $\bar{x}$ its mean, $\Delta x_t = x_{t+1} - x_t$ the first differences, and $\rho_x(k)$ its lag-$k$ autocorrelation. 
        $\Theta = \{0.01, 0.02, \dots, \max_t|x_t|\}$ enumerates positive standardized thresholds; $\tau = \min\{k\geq 1 : \rho_x(k) \leq 0\}$ is the first zero-crossing of the autocorrelation; $x_{\downarrow\tau}$ is the resulting down-sampled series; and $T\in[0,1]^{3\times 3}$ collects the empirical transition probabilities $T_{ij} = \Pr(s_{t+1}=j\mid s_t=i)$ over the equiprobable three-letter symbolization $s_t\in\{1,2,3\}$ of $x_{\downarrow\tau}$.
    } 
    \small
    \begin{tabular}{@{}llll@{}}
        \toprule
        \textbf{Type} & \textbf{Feature} & \textbf{Definition} & \textbf{Interpretation} \\
        \midrule
        & $\Lpath$ & $\sum_{i=1}^{N-1} \Delta d_i$ & Total conversational effort \\
        & $\Disp$ & $\|v_N - v_1\|_2$ & Net semantic drift \\
        $L_2$-norm & $\directness$ & $\Disp / (\Lpath + \epsilon)$ & Goal-directedness $\in (0,1]$ \\
        shape & $\meanSpeed$ & $\Lpath/(N{-}1)$ & Average speed \\
        statistics & $\stdSpeed$ & $\mathrm{std}(\Delta d)$ & Average variability \\
        & $\Vel$ & $\Disp / N$ & Directional progress per step \\ 
        & $\Circ$ & $\mathrm{std}(R) / (\bar{R} + \epsilon)$ & Tendency to revisit regions \\
        \midrule
        & $\TRA$ & $\frac{1}{N-1}\sum_{t}(\Delta x_t)^3$ & Statistical irreversibility \\
        Multivariate & $\OTpos,\; \OTneg$ & Median normalized index of $\pm$ outliers & Outlier concentration early/late \\
        temporal & $\StretchHigh$ & Longest run of consecutive $x_t > \bar{x}$ & Sustained excursion above mean \\
        (Catch22) & $\StretchDec$ & Longest run of $\Delta x_t \leq 0$ & Longest monotonic withdrawal \\
        & $\HighFluc$ & $\mathrm{tr}\,\mathrm{Cov}(T(x_{\downarrow\tau}))$ & Low-freq.\ fluctuation persistence \\
        \bottomrule
    \end{tabular}
\end{table}

When temporal features are computed on very short trajectory prefixes (e.g.\ $k{=}4$, which yields $N{=}4$ messages per trajectory), certain Catch22 statistics can be undefined: the autocorrelation sequence of a four-point series may have no zero-crossing, leaving the lag $\tau$ undefined and making the down-sampled series degenerate; likewise, the empirical transition matrix may be sparsely populated, causing missing feature values. 
Similarly, the $L_2$-norm circularity statistic $\Circ$ is explicitly missing for trajectories shorter than three steps.
These missing values are handled differently by the two classifiers.
For logistic regression, each feature is imputed with its per-feature median computed.
For gradient boosting, we use Histogram-based Gradient Boosting Classification Tree, which natively accommodates missing entries by treating them as a separate branch direction at each split, requiring no imputation.

\section{Experiments}\label{sec:experiments}
Using the features above, we design and execute a sequence of experiments that progressively isolate the geometric signal in adversarial trajectories. 
All experiments share the same pool of generated conversations: $7{,}525$ Crescendo attacks across the $4$ target LLMs and the $3$ seed benchmarks, with an overall attack-success rate of $69.6\%$ (see Appendix~\ref{app:attack_stats} for statistics of generated attacks). 
Experiments differ only in trajectory configuration and whether the number of executed turns in the conversation is exposed to the classifier. 
To evaluate the features as impact factors on the attack outcome (success or failure), we train two simple classifiers: \emph{logistic regression} (LR) with L1 regularization (strength $1.0$) as a linear classifier, and \emph{gradient boosting} (GB) (learning rate $0.1$, depth $3$ and L2 regularization $0.01$) as a tree-based ensemble method. 
Throughout, the labeled set of conversations is split $80/20$ into training and evaluation partitions. We report classification performance and \emph{surviving factors}, defined as those whose permutation importance exceeds the standard deviation of their own noise (see Appendix~\ref{app:metrics} for formulations). 
For each conversation, geometric features are extracted separately from the full conversation trajectory $\Traj$, the user-only trajectory $\TrajU$, and the assistant-only trajectory $\TrajA$, and all three feature blocks are concatenated into a single feature vector presented to the classifier. 

\subsection{Experiment 1: Conversation length confound}\label{sec:exp1}
We train both classifiers on the unfiltered set of $7{,}525$ conversation using full trajectories under two conditions: with and without number of executed turns $\Texec$ in the feature set. 
Regardless of encoder and classifier, both conditions yield near-identical and near-perfect classification, where AUROC remains at $0.991$. 
The demonstration of important factors in Figure~\ref{fig:exp1} makes the source of this invariance explicit. 
When $\Texec$ is in the feature set it is the dominant factor. 
When $\Texec$ is removed, path length $\Lpath$ takes over: it scales additively with the number of turns, so the dropped signal is reabsorbed almost exactly. 
The same pattern holds under both encoders, suggesting that the length confound is not a property of the embedding but of the geometry-to-length coupling.

\begin{figure}[t]\label{fig:exp1}
    \centering
    \includegraphics[width=\linewidth]{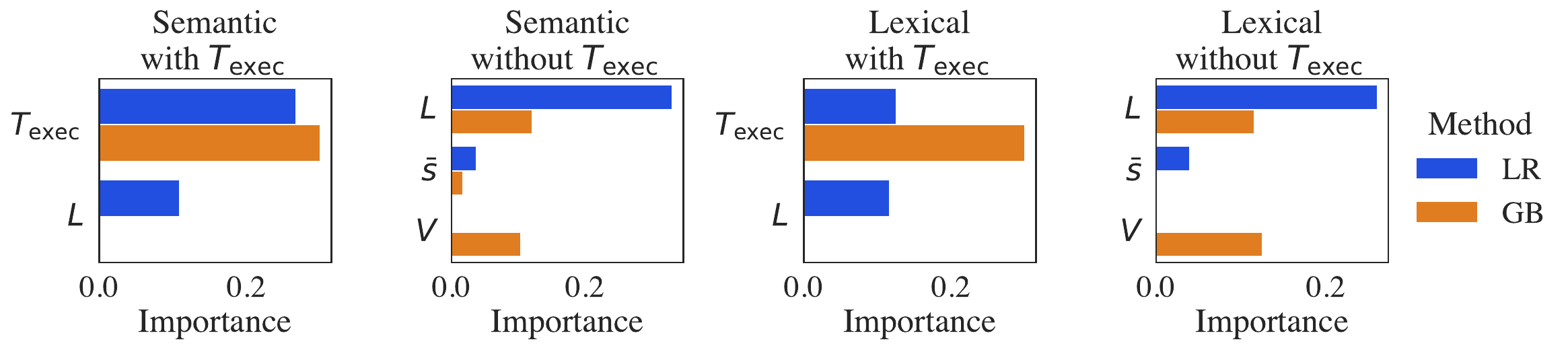}
    \caption{
        Permutation importance of top surviving factors in Experiment~1, per encoder and classifier.
        When included, $\Texec$ is the dominating factor regardless of encoder and classifier, when removed, path-length proxy $\Lpath$ inherits its role at comparable magnitude.
    }
\end{figure}

Failed attacks exhaust the $\Tmax=8$ budget by construction, while successful attacks typically terminate in $3.3$ turns on average (Appendix~\ref{app:attack_stats}). 
Any feature whose value scales with the number of steps therefore receives a free signal: this includes the path length $\Lpath$, velocity $\Vel$ and the average step pace $\meanSpeed$. 
The implication is structural rather than methodological: comparing attacks across a wide range of turn counts is intrinsically unfair to genuinely shape-based features, because length-correlated proxies dominate classification regardless of the actual trajectory geometry. 
The remaining experiments restrict to subsets of the data in which length variation is either controlled or eliminated outright. 
This length sensitivity is not specific to our setup. 
Any experimental design that compares multi-turn conversations of variable length for any purpose risks conflating trajectory geometry with turn count. 
We recommend that future work in this space report results both with and without length equalization, treating the gap between the two as a diagnostic for the severity of the confound.

\subsection{Experiment 2: Shape signals under fair length}\label{sec:exp2}
Restricting to $\TrajF{6}$ removes conversation length as a confound and leaves us with $2{,}926$ length-equalized conversations across the four generators, of which $21.9\%$ are successful escapes. 
We fit the logistic regression (LR) and gradient boosting (GB) classifiers, separately for the semantic ($\sem$) and lexical ($\lex$) encoders, repeating each run over $10$ random seeds.

Table~\ref{tab:exp2_perf} reports the seed-averaged classification performance for the four resulting configurations.
Once length information is removed, AUROC ranges from $0.653$ to $0.700$, and F1 from $0.400$ to $0.437$. 
The lexical encoder retains a small lead ($+0.03$ average AUROC and $+0.03$ average F1 over the semantic encoder), but the gap remains comparable to the seed-to-seed standard deviation observed within any single configuration ($\sigma_{\text{AUROC}}\!\in\![0.017,0.026]$, $\sigma_{\text{F1}}\!\in\![0.022,0.029]$). 
Recall stays high ($0.66$--$0.69$) while precision remains low ($0.29$--$0.32$), reflecting the class imbalance rather than any encoder-specific behavior.
We therefore conclude that, with length controlled, the choice of encoder has no statistically meaningful effect on classification performance.
 
\begin{table}[h]\label{tab:exp2_perf}
    \centering
    \caption{
        Classification performance in Experiment~2, reported as mean $\pm$ standard deviation over $10$ random seeds. 
        All conversations are length-equalized ($\TrajF{6}$, trimmed to the first six turns).
    }
    \small
    \begin{tabular}{llcccc}
        \toprule
        \textbf{Encoder} & \textbf{Classifier} & \textbf{Precision} & \textbf{Recall} & \textbf{F1} & \textbf{AUROC} \\
        \midrule
        Semantic (\sem) & LR & $0.295{\scriptstyle\pm0.020}$ & $0.675{\scriptstyle\pm0.023}$ & $0.411{\scriptstyle\pm0.022}$ & $0.672{\scriptstyle\pm0.017}$ \\
        & GB & $0.287{\scriptstyle\pm0.025}$ & $0.659{\scriptstyle\pm0.039}$ & $0.400{\scriptstyle\pm0.027}$ & $0.653{\scriptstyle\pm0.018}$ \\ 
        \midrule
        Lexical  (\lex) & LR & $0.323{\scriptstyle\pm0.025}$ & $0.682{\scriptstyle\pm0.048}$ & $0.437{\scriptstyle\pm0.029}$ & $0.700{\scriptstyle\pm0.026}$ \\ 
        & GB & $0.315{\scriptstyle\pm0.021}$ & $0.691{\scriptstyle\pm0.044}$ & $0.432{\scriptstyle\pm0.026}$ & $0.680{\scriptstyle\pm0.018}$ \\ 
        \bottomrule
    \end{tabular}
\end{table}

Table~\ref{tab:exp2_consistency} summarizes which factor families survive across configurations. 
Cell marks reflect how reliably a family appears across seeds; a family is counted as present if it survives in any of the three trajectory types ($\Traj$, $\TrajU$, $\TrajA$). 
No family appears consistently across all four configurations; the surviving signal is clearly encoder-dependent.
Under the lexical encoder, stretch-decreasing dynamics ($\StretchDec$) and negative-side outlier timing ($\OTneg$) are the dominant families, reaching \cmark\ in both LR and GB. 
Under the semantic encoder, the $L_2$-norm path descriptors ($\Lpath$) take the leading role, appearing in $8$ and $10$ of $10$ seeds for LR and GB respectively. 
$\StretchDec$ is the most broadly present family overall, reaching \cmark\ under the lexical encoder and (\cmark) under Semantic LR, but falling off under Semantic GB ($4$ seeds). 
The remaining families --- $\HighFluc$, $\StretchHigh$, $\OTpos$, and $\TRA$ --- surface intermittently, each reaching at least (\cmark) in one or more configurations. 
Circularity ($\Circ$) does not survive reliably in any configuration.

\begin{table}[h]\label{tab:exp2_consistency}
    \centering
    \caption{
        Consistency of surviving factor families across encoders and classifiers in Experiment~2 (length-equalized, aggregated over $10$ seeds). 
        Cell-level marks: \cmark\ if family present in $\geq 7$ of $10$ seeds; (\cmark) in $5$--$6$ seeds; \xmark\ otherwise. 
    }
    \small
    \begin{tabular}{lcccc} 
        \toprule
        & \multicolumn{2}{c}{\textbf{Lexical} (\lex)} & \multicolumn{2}{c}{\textbf{Semantic} (\sem)} \\
        \cmidrule(lr){2-3}\cmidrule(lr){4-5}
        \textbf{Feature family} & \textbf{LR} & \textbf{GB} & \textbf{LR} & \textbf{GB} \\
        \midrule
        $\StretchDec$       & \cmark   & \cmark   & (\cmark) & \xmark   \\
        $\HighFluc$         & \xmark   & \xmark   & (\cmark) & \xmark   \\
        \Lpath              & \xmark   & \xmark   & \cmark   & \cmark   \\
        $\OTneg$            & \cmark   & (\cmark) & \xmark   & (\cmark) \\
        $\StretchHigh$      & (\cmark) & (\cmark) & (\cmark) & \xmark   \\
        $\OTpos$            & (\cmark) & \xmark   & \xmark   & \cmark   \\
        $\TRA$              & \xmark   & \cmark   & (\cmark) & \xmark   \\
        $\Circ$             & \xmark   & \xmark   & \xmark   & \xmark   \\
        \bottomrule
    \end{tabular}
\end{table}

Comparing against the experiment in Section~\ref{sec:exp1}, where the dominant signal was carried by the turn-count proxy itself, the picture under length-equalization is qualitatively different.
The set of important features both \emph{changes} and \emph{expands}: length-correlated proxies give way to shape statistics, with stretch-decreasing dynamics ($\StretchDec$) and outlier timing ($\OTneg$) leading under the lexical encoder, and path descriptors ($\Lpath$) persisting under the semantic encoder as a genuine shape signal once length is controlled.
Encoder choice has a modest effect on performance (Lexical--Semantic AUROC gap of $\sim 0.03$, comparable to seed-to-seed noise) but a larger effect on which factor families emerge as survivors.


\subsection{Experiment 3: Early warning potential}\label{sec:exp3}
Experiment~\ref{sec:exp2} established that, once length is held fixed at six turns, a stable geometric signal of around $0.65$--$0.70$ AUROC remains. 
Whether this signal is operationally useful depends on \emph{when} it becomes available: a fingerprint that requires the entire conversation to materialize is too late to act on, while one that surfaces in the opening turns supports online intervention.
We test this directly by progressively trimming the tail of each conversation and re-fitting the same classifiers on the truncated prefix.

Concretely, we restart from the length-equalized pool of $2{,}926$ conversations used in Experiment~\ref{sec:exp2} ($\TrajF{6}$), and remove the last $k$ turns for $k\in\{0,\dots,4\}$, leaving prefixes of $6-k$ turns ($\TrajL$). 
The $k{=}0$ row therefore retains the full six-turn conversation and serves as a single-seed in-figure replay of Experiment~\ref{sec:exp2}'s Lexical+LR configuration. 
As an external reference, we additionally evaluate \texttt{llama-guard-4-12b} (Llama Guard) as a single-turn safety classifier applied to each truncated conversation; this is the dominant class of deployed runtime defense, which evaluates messages independently of trajectory shape.

Figure~\ref{fig:exp3} summarizes classification performance as a function of the retained prefix length, viewed both as ROC (which is insensitive to the class imbalance) and as precision-recall (which is not). 
Two patterns stand out. 
First, the geometric classifier degrades \emph{gracefully}: from the no-trim anchor ($k{=}0$, six retained turns) to last-turn trimmed ($k{=}1$); e.g., AUROC drops from $0.70$ to $0.67$ and PR-AUC from $0.38$ to $0.35$. 
Second, while even with only the first two turns visible ($k{=}4$) the geometric classifier remains well above chance with AUROC ${=}0.61$, \texttt{llama-guard-4-12b}'s performance drops to chance at $k = 2$ with AUROC ${=}0.49$. 

The factor structure underlying these numbers is consistent with Experiment~\ref{sec:exp2}. 
Across all horizons, the surviving features are drawn from the same families that dominated the lexical-encoder results in Experiment~\ref{sec:exp2}: stretch-decreasing dynamics ($\StretchDec$), negative-side outlier timing ($\OTneg$), and the L2-norm path descriptors ($\Lpath$, $\Disp$, $\meanSpeed$). 
At the longer prefixes ($k\in\{0,1,2\}$) the multi-variate temporal statistics, in particular $\StretchDec$ and $\OTneg$, lead the importance ranking, consistent with the lexical-encoder picture from Experiment~\ref{sec:exp2}; at the shortest prefixes ($k\in\{3,4\}$), where temporal statistics become noisier on very short series, the L2-norm path length $\Lpath$ rises to share the top of the ranking with $\OTneg$, but the AUROC stays well above chance. 
The shape of an adversarial trajectory is therefore not an artifact of the entire conversation: by the second or third turn, enough of the geometry has accumulated to discriminate the eventual outcome at AUROC well within the seed-to-seed band of the full length-equalized analysis.

These results have a direct deployment consequence. 
A monitor that consults the lexical encoder, computes a small set of length-invariant geometric features over the user--assistant exchange so far, and queries a logistic regression classifier can flag in-flight jailbreaks at AUROC $\approx 0.65$ from the first four turns of the conversation, while the dominant content-based defense slides to chance over the same prefix. 
The signal that distinguishes successful from failed multi-turn attacks is front-loaded, encoder-light, and available before the attack itself completes.

\begin{figure}[t]\label{fig:exp3}
    \centering
    \includegraphics[width=\linewidth]{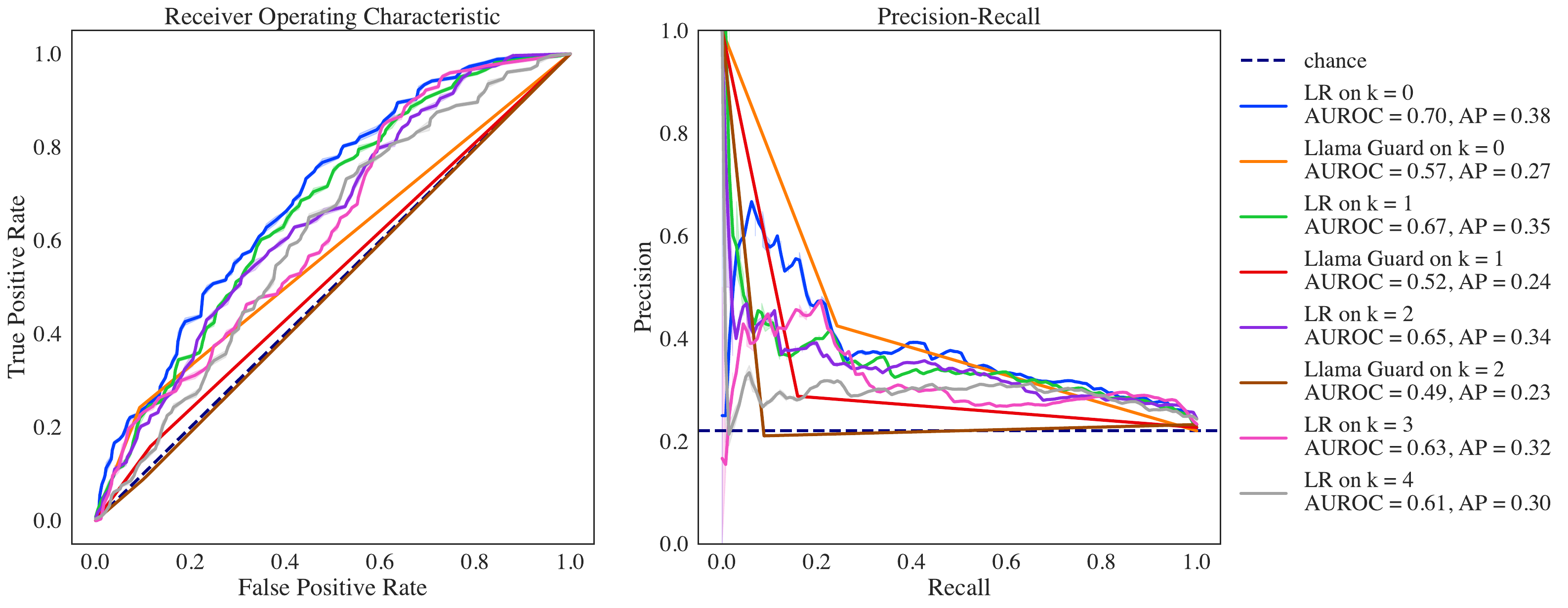}
    \caption{
        Receiver-operating-characteristic (left) and precision-recall (right) curves of the geometric LR classifier (lexical encoder, length-equalized prefixes) and the content-based \texttt{llama-guard-4-12b} (LG) baseline, across trim levels $k\in\{0,\dots,4\}$, where the retained prefix has length $6{-}k$. 
        Dashed navy lines mark the chance reference: the diagonal in the ROC panel and the positive-class base rate ($21.9\%$) in the PR panel. 
        The full per-metric breakdown is given in Appendix~\ref{app:exp3} (Table~\ref{tab:exp3_perf}).
    }
\end{figure}

\section{Theoretical analysis}\label{sec:theory}
The experiments in the previous section reveal three structural regularities: length dominates na\"ive classification (Experiment~\ref{sec:exp1}), encoder choice is irrelevant once length is held fixed (Experiment~\ref{sec:exp2}), and the residual shape signal that persists after length equalization is front-loaded in time (Experiment~\ref{sec:exp3}). 
In this section, we provide formal results that explain each regularity and yield testable quantitative predictions. 
Throughout, we model the binary label $Y \in \{0,1\}$ (failed/successful attack) and a feature vector $\mathbf{f} \in \mathbb{R}^p$ extracted from the conversation trajectory. 
All proofs are in Appendix~\ref{app:proofs}.

\paragraph{Confounding effect of conversation length}
We first formalize the observation that length-correlated features dominate classification and predict the performance drop when length is controlled.

\begin{definition}
    Let $\ell(\mathbf{f})$ denote the \emph{length component} of the feature vector, a subvector of features whose population means scale monotonically with the number of executed turns $\Texec$, and $\mathbf{g}(\mathbf{f})$ denote the \emph{shape component}, the subvector of features whose population means are invariant to $\Texec$ when computed on length-equalized trajectories. 
    Write $\mathbf{f} = (\ell, \mathbf{g})$.
\end{definition}

\begin{assumption}\label{as:logistic}
    The class-conditional distributions satisfy a Gaussian location model: for $y \in \{0,1\}$, $ \ell \mid Y{=}y \;\sim\; \mathcal{N}(\mu_\ell^{(y)},\, \sigma_\ell^2), \mathbf{g} \mid Y{=}y \;\sim\; \mathcal{N}(\boldsymbol{\mu}_g^{(y)},\, \sigma_g^2 I_q)$, with $\ell \perp \mathbf{g} \mid Y$ (conditional independence of length and shape components given the label).
\end{assumption}

The Gaussian location model is standard in signal detection theory \cite{green1966signal} and is the implicit generative model behind Fisher's linear discriminant, which logistic regression approximates when class-conditional distributions share a common covariance. 
Since both classifiers in our experiments produce decision boundaries that are monotone functions of a linear projection under this model, the assumption lets us derive closed-form expressions for discriminability that would otherwise require nonparametric estimation.

\begin{proposition}\label{prop:auroc}
    Under Assumption~\ref{as:logistic}, the Bayes-optimal AUROC on the full feature vector
    $\mathbf{f}$ is
    \begin{equation*}
      \mathrm{AUROC}_{\mathrm{full}}
      \;=\;
      \Phi\!\left(
        \sqrt{\frac{\mathrm{SNR}_\ell^2 + \mathrm{SNR}_g^2}{2}}
      \right),
    \end{equation*}
    where $\mathrm{SNR}_\ell := \sigma_\ell^{-1} |\mu_\ell^{(1)} - \mu_\ell^{(0)}|, \mathrm{SNR}_g := \sigma_g^{-1} \|\boldsymbol{\mu}_g^{(1)} - \boldsymbol{\mu}_g^{(0)}\|$, and $\Phi$ is the standard normal CDF. 
    On length-equalized data ($\mathrm{SNR}_\ell = 0$), the AUROC reduces to $\mathrm{AUROC}_{\mathrm{shape}} = \Phi\!\bigl(\mathrm{SNR}_g / \sqrt{2}\bigr)$.
\end{proposition}

From Experiment~\ref{sec:exp1}, $\mathrm{AUROC}_{\mathrm{full}} \approx 0.991$, so $\Phi^{-1}(0.991) \approx 2.37$ and $d'_{\mathrm{full}} = \sqrt{2}\cdot 2.37 \approx 3.35$.
From Table~\ref{tab:exp2_perf}, $\mathrm{AUROC}_{\mathrm{shape}} \approx 0.700$ (best length-equalized configuration), giving $d'_{\mathrm{shape}} = \sqrt{2}\,\Phi^{-1}(0.700) \approx 0.74$.
The decomposition predicts $d'_{\mathrm{full}} = \sqrt{(d'_\ell)^2 + (d'_{\mathrm{shape}})^2}$, which yields $d'_\ell \approx \sqrt{3.35^2 - 0.74^2} \approx 3.27$. 
The length signal thus accounts for $(3.27/3.35)^2 \approx 95.1\%$ of the squared discriminability, consistent with the experimental observation that removing length information causes a dramatic AUROC drop while the residual shape signal, though modest in absolute terms, remains reliably above chance.


\paragraph{Minimum prefix length for reliable profiling}
Experiment~\ref{sec:exp3} provides us with a number of geometric features that are informative for early detection of attacks. 
Given such a feature whose class-conditional distributions are separated, it is a natural question to ask: how many turns of observation are needed for reliable detection?

\begin{assumption}\label{as:feature}
    Let $h_k$ denote the value of a scalar geometric feature computed from the first $k$ turns of the trajectory. 
    For each class $y \in \{0,1\}$, $h_k \mid Y{=}y$ is sub-Gaussian with mean $\mu_h^{(y)}(k)$ and variance proxy $\sigma_h^2$ uniformly in $k$. 
    The population separation satisfies $\Delta(k) := |\mu_h^{(1)}(k) - \mu_h^{(0)}(k)| \geq \Delta_0 > 0$ for all $k \geq k_0$, i.e.\ the feature becomes separated after a warm-up period $k_0$.
\end{assumption}

Our geometric features are computed from trajectories in a bounded embedding space, so their range is mechanically bounded and sub-Gaussianity holds with a parameter that depends on this range. 
The separation condition formalizes the empirical observation from Experiment~\ref{sec:exp2} that certain features (e.g.\ $\StretchDec$, $\OTneg$) maintain a nonzero gap between class means even after length equalization, and the warm-up period $k_0$ allows for features that require a minimum number of turns before they become well-defined.

\begin{proposition}\label{prop:prefix}
    Under Assumption~\ref{as:feature}, the Bayes error of a classifier on $h_k$satisfies
    \begin{equation*}
      P_{\mathrm{err}}(k)
      \;\leq\;
      \exp\!\left(-\frac{\Delta(k)^2}{8\,\sigma_h^2}\right)
      \quad \text{for all } k \geq k_0.
    \end{equation*}
    Consequently, the minimum prefix length for error at most $\varepsilon$ is
    \begin{equation*}
    k^*(\varepsilon)
    \;=\;
    \min\!\Bigl\{k \geq k_0 : \Delta(k) \geq \sigma_h\sqrt{8\log(\varepsilon)}\Bigr\}.    
    \end{equation*}
    If $\Delta(k) \geq \Delta_0$ for all $k \geq k_0$, then
    $k^*(\varepsilon) = k_0$ whenever
    $\Delta_0 \geq \sigma_h\sqrt{8\log(\varepsilon)}$.
\end{proposition}

The stability of AUROC across trim levels $k \in \{0,\ldots,4\}$ in Experiment~\ref{sec:exp3} (Table~\ref{tab:exp3_perf}) implies that the discriminative features attain their separation $\Delta_0$ very early --- effectively within one or two turns of observation. 
This is precisely the regime in which Proposition~\ref{prop:prefix} predicts $k^*(\varepsilon) = k_0$: the feature separation is large enough relative to within-class variance that even a short observed prefix suffices.
The proposition also clarifies why \emph{different} features can have very different warm-up periods: features with larger $\Delta_0/\sigma_h$ ratios (such as $\Lpath$ at short prefixes in Experiment~\ref{sec:exp3}) become reliable earlier, while features with smaller ratios require longer prefixes to overcome noise.

\paragraph{Representation invariance}
Experiment~\ref{sec:exp2} shows that a sparse lexical encoder and a dense semantic encoder yield near-identical classification. 
We give a sufficient condition for this to happen.

\begin{definition}
    Two encoders $\phi_1, \phi_2 : \mathcal{M} \to \mathbb{R}^d$ are
    \emph{rank-preserving} with respect to a feature functional
    $F : (\mathbb{R}^d)^n \to \mathbb{R}$ if, for all pairs of trajectories
    $\mathcal{T}, \mathcal{T}'$ of equal length,
    \begin{equation*}
        F(\phi_1(\mathcal{T})) \;\geq\; F(\phi_1(\mathcal{T}'))
        \quad\Longleftrightarrow\quad
        F(\phi_2(\mathcal{T})) \;\geq\; F(\phi_2(\mathcal{T}')).  
    \end{equation*}
\end{definition}

\begin{proposition}\label{prop:invariance}
    If two encoders are rank-preserving with respect to every feature functional in the set $\{F_1, \ldots, F_p\}$, then any classifier whose decisions depend only on univariate feature ranks has identical AUROC under both encoders. For linear classifiers, identical AUROC additionally requires that the within-class covariance structure of the feature vector is shared across encoders up to a global scalar.
\end{proposition}

The near-perfect AUROC agreement between $\lex$ and $\sem$ in Table~\ref{tab:exp2_perf} suggests that the two encoders are approximately rank-preserving for the dominant features (path length, velocity, and the surviving Catch22 features). 
This is consistent with the Platonic Representation Hypothesis~\citep{platonic}, which posits that diverse representation learners converge toward a shared statistical model of reality. If the underlying rank structure of trajectory geometries is a property of the data rather than the encoder, rank preservation follows naturally.
For gradient boosting, Proposition~\ref{prop:invariance} guarantees that this suffices for AUROC invariance. For logistic regression, the proposition requires the additional condition that the within-class covariance structure is shared across encoders; the near-identical performance of LR across encoders in Table~\ref{tab:exp2_perf} suggests this stronger condition holds approximately.
The $L_2$-norm features ($\Lpath$, $\Disp$, $\directness$, $\meanSpeed$, $\Vel$, $\Circ$) depend on pairwise distances in the embedding space, and both encoders preserve the relative ordering of which trajectory pairs are more or less geometrically ``extreme''. 
The result also predicts that encoder invariance should \emph{break} for features that depend on absolute scale rather than rank ordering. 
We observe this in Table~\ref{tab:exp2_consistency}, where the surviving Catch22 families differ substantially across encoders: $\StretchDec$ and $\OTneg$ dominate under the lexical encoder while $\Lpath$ leads under the semantic encoder.

\section{Discussion and Conclusion}\label{sec:conclusion}
The most striking finding is how little of the na\"ive classification signal survives once conversation length is controlled: the AUROC drops from $0.99$ to roughly $0.65$--$0.70$, and our theoretical decomposition attributes $95.1\%$ of the squared discriminability to length alone. 
This is a cautionary result for the broader conversation-analysis literature, where trajectory-level features are increasingly used as proxies for dialogue quality~\citep{gooding2025interactiondynamicsrewardsignal} or alignment~\citep{simhi2026oldhabitsdiehard}. 
Any study that compares conversations of unequal length risks mistaking a length artifact for a structural signal.

The residual geometric signal, while modest, has properties that make it practically interesting. 
Classification performance is similar across encoders, which means a deployment need not commit to an expensive embedding model. 
It is front-loaded, appearing within the first two to three turns, which is the window in which intervention can still prevent harm. 
And it is carried by interpretable features: monotonic withdrawal in latent dimensions (the target model gradually conceding), concentrated outlier timing (sharp shifts at specific turns), and persistent low-frequency fluctuation. 
These are not opaque classifier artifacts; they correspond to recognizable dynamics of how an attacker incrementally steers a conversation.

The baseline comparisons in Experiment~\ref{sec:exp3} underscore a structural limitation of content-level defenses. 
The performance of \texttt{llama-guard-4-12b} collapses to chance when the final turns, the ones most likely to contain overt harmful content, are withheld. 
This is by design, because content filters are intended to detect harm, not intent. 
A trajectory-level monitor like ours would operate on a complementary signal and is therefore most valuable precisely where content filters are weakest, that is, early in the conversation before harmful content has been produced.

\paragraph{Limitations.}
Four limitations should be noted. 
\emph{(i)}~All experiments use a single attack strategy (Crescendo); whether the geometric fingerprint generalizes to other multi-turn methods (e.g.\ TAP, repeated-question attacks) remains open. 
\emph{(ii)}~The $\Tmax=8$ turn budget shapes the length confound; a different cap would shift both the success-rate distribution and the relative importance of length-sensitive features. 
\emph{(iii)}~Attack labels are assigned by a single LLM-as-judge with a fixed threshold; sensitivity to the judge or threshold is unexplored. 
\emph{(iv)}~Encoder invariance is established along a coarse sparse-vs.-dense axis; whether two dense encoders of different capacity agree as closely is not tested.

\paragraph{Future work.}
The immediate next step is cross-attack generalization: does a classifier trained on Crescendo trajectories transfer to TAP or repeated-question attacks, or does each strategy leave a distinct geometric signature? 
This would test whether the fingerprint reflects attacker behavior or model-specific response patterns. 
Beyond classification, replacing hard labels with calibrated trajectory-level probabilities would enable integration with existing guardrail pipelines, where a geometric risk score could modulate the sensitivity of content-level filters as a conversation progresses.

\begin{ack}
This research is supported by the Indian Institute of Management Bangalore Young Faculty Research Grant. SM acknowledges the creators of the anime Psycho-Pass for being the source of inspiration for this work.
\end{ack}

\bibliographystyle{abbrvnat}
\bibliography{reference}

@InProceedings{platonic,
  title = 	 {Position: The Platonic Representation Hypothesis},
  author =       {Huh, Minyoung and Cheung, Brian and Wang, Tongzhou and Isola, Phillip},
  booktitle = 	 {Proceedings of the 41st International Conference on Machine Learning},
  pages = 	 {20617--20642},
  year = 	 {2024},
  volume = 	 {235},
  series = 	 {Proceedings of Machine Learning Research},
  month = 	 {21--27 Jul},
  publisher =    {PMLR},
}

@misc{prompt-guard-86m,
  title = {Model Card - Prompt Guard},
  author = {Meta},
  year = {2024},
  url = {https://huggingface.co/meta-llama/Prompt-Guard-86M},
  note = {Accessed: Apr 30, 2026}
}

@misc{vijil-mbert,
  title = {Model Card for Vijil Prompt Injection},
  author = {Vijil},
  year = {2025},
  url = {https://huggingface.co/vijil/mbert-prompt-injection},
  note = {Accessed: Apr 30, 2026}
}

@misc{hackett2025bypassingllmguardrailsempirical,
      title={Bypassing LLM Guardrails: An Empirical Analysis of Evasion Attacks against Prompt Injection and Jailbreak Detection Systems}, 
      author={William Hackett and Lewis Birch and Stefan Trawicki and Neeraj Suri and Peter Garraghan},
      year={2025},
      eprint={2504.11168},
      archivePrefix={arXiv},
      primaryClass={cs.CR},
      url={https://arxiv.org/abs/2504.11168}, 
}

@misc{nova,
  title = {NOVA: The Prompt Pattern Matching},
  author = {Thomas Roccia},
  year = {2025},
  url = {https://github.com/Nova-Hunting/nova-framework},
  note = {Accessed: Apr 30, 2026}
}

@misc{fath,
      title={FATH: Authentication-based Test-time Defense against Indirect Prompt Injection Attacks}, 
      author={Jiongxiao Wang and Fangzhou Wu and Wendi Li and Jinsheng Pan and Edward Suh and Z. Morley Mao and Muhao Chen and Chaowei Xiao},
      year={2024},
      eprint={2410.21492},
      archivePrefix={arXiv},
      primaryClass={cs.CR},
      url={https://arxiv.org/abs/2410.21492}, 
}

@misc{spotlighting,
      title={Defending Against Indirect Prompt Injection Attacks With Spotlighting}, 
      author={Keegan Hines and Gary Lopez and Matthew Hall and Federico Zarfati and Yonatan Zunger and Emre Kiciman},
      year={2024},
      eprint={2403.14720},
      archivePrefix={arXiv},
      primaryClass={cs.CR},
      url={https://arxiv.org/abs/2403.14720}, 
}

@misc{ayub2024embeddingbasedclassifiersdetectprompt,
      title={Embedding-based classifiers can detect prompt injection attacks}, 
      author={Md. Ahsan Ayub and Subhabrata Majumdar},
      year={2024},
      eprint={2410.22284},
      archivePrefix={arXiv},
      primaryClass={cs.CR},
      url={https://arxiv.org/abs/2410.22284}, 
}

@misc{granite-guardian-41,
  title = {Introducing the IBM Granite 4.1 family of models},
  author = {IBM},
  year = {2026},
  url = {https://research.ibm.com/blog/granite-4-1-ai-foundation-models},
  note = {Accessed: Apr 30, 2026}
}

@misc{llama-guard-4-12b,
  title = {Llama Guard 4: Natively Multimodal Safeguard Model},
  author = {Meta},
  year = {2025},
  url = {https://huggingface.co/meta-llama/Llama-Guard-4-12B},
  note = {Accessed: Apr 30, 2026}
}

@misc{derczynski2024garakframeworksecurityprobing,
      title={garak: A Framework for Security Probing Large Language Models}, 
      author={Leon Derczynski and Erick Galinkin and Jeffrey Martin and Subho Majumdar and Nanna Inie},
      year={2024},
      eprint={2406.11036},
      archivePrefix={arXiv},
      primaryClass={cs.CL},
      url={https://arxiv.org/abs/2406.11036}, 
}

@misc{liu2024autodangeneratingstealthyjailbreak,
      title={AutoDAN: Generating Stealthy Jailbreak Prompts on Aligned Large Language Models}, 
      author={Xiaogeng Liu and Nan Xu and Muhao Chen and Chaowei Xiao},
      year={2024},
      eprint={2310.04451},
      archivePrefix={arXiv},
      primaryClass={cs.CL},
      url={https://arxiv.org/abs/2310.04451}, 
}

@misc{leviathan2025promptrepetitionimprovesnonreasoning,
      title={Prompt Repetition Improves Non-Reasoning LLMs}, 
      author={Yaniv Leviathan and Matan Kalman and Yossi Matias},
      year={2025},
      eprint={2512.14982},
      archivePrefix={arXiv},
      primaryClass={cs.LG},
      url={https://arxiv.org/abs/2512.14982}, 
}

@article{mehrotra2024tap,
  title={Tree of Attacks: Jailbreaking Black-Box LLMs with Crafted Prompts},
  author={Mehrotra, Anay and Zampetakis, Manolis and Kassianik, Paul and Nelson, Blaine and Anderson, Hyrum and Singer, Yaron and Karbasi, Amin},
  journal={arXiv preprint arXiv:2312.02119},
  year={2024}
}

@book{green1966signal,
  author    = {Green, David M. and Swets, John A.},
  title     = {Signal Detection Theory and Psychophysics},
  publisher = {Wiley},
  year      = {1966},
  address   = {New York}
}

@misc{laban2025llmslostmultiturnconversation,
      title={LLMs Get Lost In Multi-Turn Conversation}, 
      author={Philippe Laban and Hiroaki Hayashi and Yingbo Zhou and Jennifer Neville},
      year={2025},
      eprint={2505.06120},
      archivePrefix={arXiv},
      primaryClass={cs.CL},
      url={https://arxiv.org/abs/2505.06120}, 
}

@misc{munoz2024pyritframeworksecurityrisk,
      title={PyRIT: A Framework for Security Risk Identification and Red Teaming in Generative AI Systems},
      author={Gary D. Lopez Munoz and Amanda J. Minnich and Roman Lutz and Richard Lundeen and Raja Sekhar Rao Dheekonda and Nina Chikanov and Bolor-Erdene Jagdagdorj and Martin Pouliot and Shiven Chawla and Whitney Maxwell and Blake Bullwinkel and Katherine Pratt and Joris de Gruyter and Charlotte Siska and Pete Bryan and Tori Westerhoff and Chang Kawaguchi and Christian Seifert and Ram Shankar Siva Kumar and Yonatan Zunger},
      year={2024},
      eprint={2410.02828},
      archivePrefix={arXiv},
      primaryClass={cs.CR},
      url={https://arxiv.org/abs/2410.02828},
}

@misc{russinovich2025greatwritearticlethat,
      title={Great, Now Write an Article About That: The Crescendo Multi-Turn LLM Jailbreak Attack}, 
      author={Mark Russinovich and Ahmed Salem and Ronen Eldan},
      year={2025},
      eprint={2404.01833},
      archivePrefix={arXiv},
      primaryClass={cs.CR},
      url={https://arxiv.org/abs/2404.01833}, 
}

@misc{lubba2019catch22canonicaltimeseriescharacteristics,
      title={catch22: CAnonical Time-series CHaracteristics}, 
      author={Carl H Lubba and Sarab S Sethi and Philip Knaute and Simon R Schultz and Ben D Fulcher and Nick S Jones},
      year={2019},
      eprint={1901.10200},
      archivePrefix={arXiv},
      primaryClass={cs.IR},
      url={https://arxiv.org/abs/1901.10200}, 
}

@misc{mazeika2024harmbenchstandardizedevaluationframework,
      title={HarmBench: A Standardized Evaluation Framework for Automated Red Teaming and Robust Refusal}, 
      author={Mantas Mazeika and Long Phan and Xuwang Yin and Andy Zou and Zifan Wang and Norman Mu and Elham Sakhaee and Nathaniel Li and Steven Basart and Bo Li and David Forsyth and Dan Hendrycks},
      year={2024},
      eprint={2402.04249},
      archivePrefix={arXiv},
      primaryClass={cs.LG},
      url={https://arxiv.org/abs/2402.04249}, 
}

@misc{zou2023universaltransferableadversarialattacks,
      title={Universal and Transferable Adversarial Attacks on Aligned Language Models}, 
      author={Andy Zou and Zifan Wang and Nicholas Carlini and Milad Nasr and J. Zico Kolter and Matt Fredrikson},
      year={2023},
      eprint={2307.15043},
      archivePrefix={arXiv},
      primaryClass={cs.CL},
      url={https://arxiv.org/abs/2307.15043}, 
}

@misc{gooding2025interactiondynamicsrewardsignal,
      title={Interaction Dynamics as a Reward Signal for LLMs}, 
      author={Sian Gooding and Edward Grefenstette},
      year={2025},
      eprint={2511.08394},
      archivePrefix={arXiv},
      primaryClass={cs.CL},
      url={https://arxiv.org/abs/2511.08394}, 
}

@misc{simhi2026oldhabitsdiehard,
      title={Old Habits Die Hard: How Conversational History Geometrically Traps LLMs}, 
      author={Adi Simhi and Fazl Barez and Martin Tutek and Yonatan Belinkov and Shay B. Cohen},
      year={2026},
      eprint={2603.03308},
      archivePrefix={arXiv},
      primaryClass={cs.CL},
      url={https://arxiv.org/abs/2603.03308}, 
}

\appendix
\section{Crescendo attack mechanism}\label{app:crescendo}
Crescendo~\citep{russinovich2025greatwritearticlethat} is a multi-turn jailbreak strategy that escalates toward a forbidden objective across several turns rather than encoding it in a single prompt. 
Its premise is that single-turn safety filters score each message in isolation, so a request that would be refused outright can often be elicited indirectly: the attacker first establishes a conversational context that makes the harmful request look like a natural continuation, and only then closes the loop. 
We summarize the mechanism here; the full algorithm and prompt templates are described in the original paper.

\paragraph{Escalation.}
At each turn, the adversarial bot produces a prompt that takes the objective bot's previous response as a foothold and pushes one increment toward the seed objective---moving, for example, from background context, to specific details, to the explicit forbidden request. 
Each step stays within the conversational frame already established; this accumulating frame is what gives the attack its name. 
The seed objective itself is never revealed verbatim until the surrounding context has been primed enough that the objective bot is likely to comply.

\paragraph{Scoring and termination.}
After each assistant turn, the scoring bot (LLM-as-judge) rates how completely the response has satisfied the seed objective on a $[0,1]$ scale. 
The attack is declared a \emph{success} the first time this score crosses the threshold ($0.8$ in our setup); otherwise it is a \emph{failure} once the turn budget $\Tmax=8$ is exhausted. 
Each turn comprises one user message followed by one assistant message.

\paragraph{Backtracking.}
When the objective bot refuses or produces a non-compliant response, the adversarial bot is allowed to \emph{backtrack}: it discards the most recent user--assistant exchange and retries the same conversational position with a rephrased prompt. 
This lets the attacker route around individual refusals without resetting the entire dialogue. 
The number of backtracks per attack is capped at $2$. 
The trajectory $\Traj$ that Sections~\ref{sec:methodology}--\ref{sec:experiments} encode and analyze is the conversation that remains after the orchestrator commits to a single forward path---i.e., backtracked branches are pruned from the recorded turns.

\paragraph{Implementation.} 
We use the Crescendo attack executor from PyRIT~\citep{munoz2024pyritframeworksecurityrisk}, which implements the loop above with the prompt templates released alongside the original paper. 
Bot identities and decoding parameters are listed in Section~\ref{sec:methodology}.

\section{Attack generation statistics}\label{app:attack_stats}
\begin{figure}[h]\label{fig:attack_stats}
    \centering
    \includegraphics[width=\linewidth]{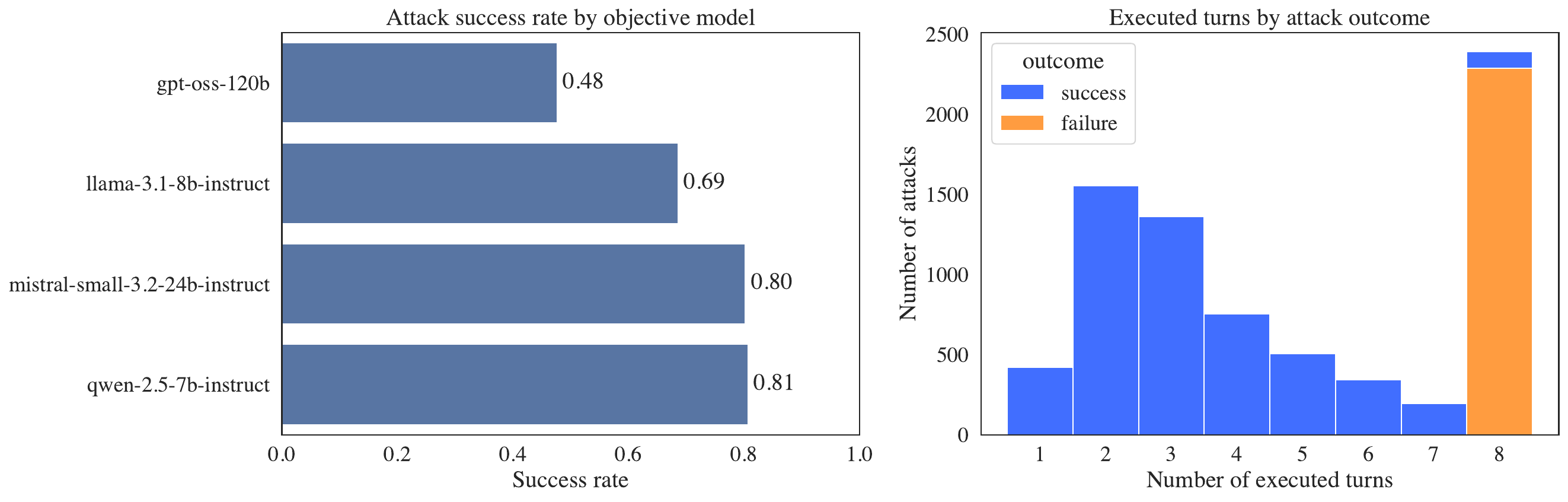}
    \caption{
        Aggregate statistics of the Crescendo{} attack pool used throughout the paper. 
        \emph{Left:} attack success rate per objective LLM (sorted ascending). 
        \emph{Right:} empirical distribution of the number of executed turns per attack, stacked by outcome; failures saturate the $\Tmax{=}8$ budget by construction while successes are concentrated in the first three to four turns.
    }
\end{figure}

Figure~\ref{fig:attack_stats} summarizes the conversational dataset that all three experiments share: $7{,}525$ Crescendo attacks executed against four target LLMs, with each attack capped at $\Tmax=8$ turns and $2$ backtracks. 
Attack difficulty varies substantially by target: \texttt{gpt-oss-120b} resists the largest fraction ($48\%$ success rate), \texttt{mistral-small-3.2-24b-instruct} and \texttt{qwen-2.5-7b-instruct} are the most permissive ($80\%$ and $81\%$), and \texttt{llama-3.1-8b-instruct}---which doubles as the adversarial bot, so this is the within-family setting where adversary and target share a base checkpoint---sits in between at $69\%$. 
The aggregate success rate over the full pool, $69.6\%$, is the value quoted in Section~\ref{sec:experiments}.

The right panel underwrites the length-confound diagnosis of Experiment~\ref{sec:exp1}. 
By construction of the attack strategy, every \emph{failed} attack exhausts the turn budget exactly, so the failure histogram is a delta at $\Tmax=8$; \emph{successful} attacks terminate strictly faster, with a mean of $3.3$ turns and a long thin tail. 
Any feature whose value scales with conversation length therefore inherits a near-perfect signal from this near-deterministic coupling between outcome and turn count, which is precisely the artifact Experiment~\ref{sec:exp1} sets out to isolate.

\section{Classification metrics}\label{app:metrics}
We report five standard binary-classification metrics throughout the paper. 
Let $y_i\in\{0,1\}$ be the true label of conversation $i$ (with $1$ denoting a successful attack), $s_i\in[0,1]$ a continuous score produced by the classifier, and $\hat{y}_i = \mathbf{1}[s_i \geq 0.5]$ the corresponding hard prediction at the default $0.5$ threshold. 
The four entries of the confusion matrix are
\begin{align*}
    \mathrm{TP} &= \#\{i : y_i = 1,\, \hat{y}_i = 1\}, & \mathrm{FP} &= \#\{i : y_i = 0,\, \hat{y}_i = 1\}, \\
    \mathrm{TN} &= \#\{i : y_i = 0,\, \hat{y}_i = 0\}, & \mathrm{FN} &= \#\{i : y_i = 1,\, \hat{y}_i = 0\}.
\end{align*}

\paragraph{Threshold-dependent metrics.} 
The first three metrics are evaluated at the default threshold ($0.5$) and characterize the classifier at a single operating point:
\begin{align*}
    \text{Precision} &= \frac{\mathrm{TP}}{\mathrm{TP}+\mathrm{FP}}, &
    \text{Recall} &= \frac{\mathrm{TP}}{\mathrm{TP}+\mathrm{FN}}, &
    \text{F1} &= \frac{2\,\text{Precision}\cdot\text{Recall}}{\text{Precision}+\text{Recall}}.
\end{align*}

\paragraph{Threshold-free metrics.} 
The remaining two summarize the entire score-rank ordering and are independent of any threshold choice. 
The Area Under the ROC curve is most cleanly stated as a rank statistic,
\begin{equation*}
    \text{AUROC} \;=\; \Pr\bigl(s_i > s_j \;\big|\; y_i = 1,\, y_j = 0\bigr),
\end{equation*}
i.e.\ the probability that a randomly drawn positive instance receives a higher score than a randomly drawn negative one; it is invariant to monotone score transformations, insensitive to class imbalance, and equal to $\tfrac{1}{2}$ for a random classifier. 
Average Precision (AP), the area under the precision-recall curve, is computed as a finite sum over the unique score thresholds:
\begin{equation*}
    \text{AP} \;=\; \sum_{n}\bigl(\mathrm{Recall}_n - \mathrm{Recall}_{n-1}\bigr)\,\mathrm{Precision}_n,
\end{equation*}
where the index $n$ runs over the sorted unique values of $s_i$. 
Unlike AUROC, AP is sensitive to the positive-class base rate $\pi$: a constant predictor achieves $\text{AP} = \pi$, so AP values should be read against this floor rather than against the $0.5$ AUROC baseline.

\section{Experiment 3: Classification performance}\label{app:exp3}
\begin{table}[h]\label{tab:exp3_perf}
    \centering
    \caption{
        Classification performance in Experiment~3 as a function of retained prefix length. 
        Conversations are first length-equalized to six turns ($\TrajF{6}$) and then truncated by removing the last $k$ turns, leaving $6{-}k$ turns; $k{=}0$ retains the full six-turn conversation. 
        \texttt{llama-guard-4-12b} (LG) is included as a content-based runtime safety baseline. 
        All entries use the lexical encoder $\lex$.
    }
    \small
    \begin{tabular}{cclcccc}
        \toprule
        $k$ & Prefix & \textbf{Method} & \textbf{Precision} & \textbf{Recall} & \textbf{F1} & \textbf{AUROC} \\
        \midrule
        $0$ & $6$ & LR & $0.318$ & $0.721$ & $0.442$ & $0.697$ \\
            &     & GB & $0.317$ & $0.705$ & $0.438$ & $0.682$ \\
            &     & LG & $0.425$ & $0.242$ & $0.308$ & $0.575$ \\
        \midrule
        $1$ & $5$ & LR & $0.316$ & $0.636$ & $0.422$ & $0.668$ \\
            &     & GB & $0.315$ & $0.758$ & $0.445$ & $0.659$ \\
            &     & LG & $0.288$ & $0.159$ & $0.205$ & $0.522$ \\
        \midrule
        $2$ & $4$ & LR & $0.315$ & $0.603$ & $0.414$ & $0.648$ \\
            &     & GB & $0.286$ & $0.603$ & $0.388$ & $0.622$ \\
            &     & LG & $0.211$ & $0.088$ & $0.124$ & $0.494$ \\
        \midrule
        $3$ & $3$ & LR & $0.266$ & $0.593$ & $0.367$ & $0.630$ \\
            &     & GB & $0.319$ & $0.674$ & $0.433$ & $0.651$ \\
            &     & LG & $0.280$ & $0.104$ & $0.151$ & $0.512$ \\
        \midrule
        $4$ & $2$ & LR & $0.307$ & $0.593$ & $0.405$ & $0.609$ \\
            &     & GB & $0.275$ & $0.550$ & $0.367$ & $0.586$ \\
            &     & LG & $0.167$ & $0.043$ & $0.068$ & $0.488$ \\
        \bottomrule
    \end{tabular}
\end{table}

Table~\ref{tab:exp3_perf} reports the per-metric classification performance behind the ROC and PR curves shown in Figure~\ref{fig:exp3}. 
For each trim level $k\in\{0,\dots,4\}$ (with retained prefix of $6-k$ turns; $k{=}0$ is a single-seed replay of Experiment~\ref{sec:exp2}'s lexical+LR setup) we list precision, recall, F1, and AUROC for the geometric logistic-regression and gradient-boosting classifiers (LR, GB) and for the \texttt{llama-guard-4-12b} content baseline (LG). 
Llama Guard's F1 column exposes the actual single-turn detection quality, which decays toward zero as the prefix shortens, while the geometric classifiers retain F1 around $0.4$ and AUROC well above chance at every horizon.

\section{Proofs of Theoretical Results}\label{app:proofs}
\begin{proof}[Proof of Proposition~\ref{prop:auroc}]
    Under the Gaussian location model with equal covariances, the Bayes-optimal decision rule is Fisher's linear discriminant, whose projection yields the scalar
    \begin{equation*}
      Z \;=\; \frac{\mu_\ell^{(1)}-\mu_\ell^{(0)}}{\sigma_\ell^2}\,\ell
      \;+\;
      \frac{(\boldsymbol{\mu}_g^{(1)}-\boldsymbol{\mu}_g^{(0)})^\top}{\sigma_g^2}\,\mathbf{g}.
    \end{equation*}
    By conditional independence, $\mathrm{Var}(Z \mid Y) = (\mu_\ell^{(1)}-\mu_\ell^{(0)})^2/\sigma_\ell^2 + \|\boldsymbol{\mu}_g^{(1)}-\boldsymbol{\mu}_g^{(0)}\|^2/\sigma_g^2 = \mathrm{SNR}_\ell^2 + \mathrm{SNR}_g^2$.
    Let $Z_1 \sim Z \mid Y{=}1$ and $Z_0 \sim Z \mid Y{=}0$ be independent draws from the two class-conditional distributions. 
    Their difference $Z_1 - Z_0$ is Gaussian with mean $d' \cdot \sigma_Z$ and variance $2\sigma_Z^2$, where $d' = \sqrt{\mathrm{SNR}_\ell^2 + \mathrm{SNR}_g^2}$ and $\sigma_Z^2 = \mathrm{Var}(Z \mid Y)$. The AUROC equals
    \begin{equation*}
        \mathrm{AUROC} = P(Z_1 > Z_0) = \Phi\!\left(\frac{d'}{\sqrt{2}}\right),
    \end{equation*}
    which is the standard result for equal-variance Gaussian discriminants~\citep[see, e.g.,][]{green1966signal}.
    Setting $\mathrm{SNR}_\ell = 0$ yields the length-equalized case: $\mathrm{AUROC}_{\mathrm{shape}} = \Phi\!\bigl(\mathrm{SNR}_g / \sqrt{2}\bigr)$.
\end{proof}

\begin{proof}[Proof of Proposition~\ref{prop:prefix}]
    For a classifier at the midpoint $t = (\mu_h^{(0)}(k) + \mu_h^{(1)}(k))/2$, an error under class $y$ requires $h_k$ to deviate from its class mean by at least $\Delta(k)/2$ in the direction of the opposing class. 
    This is a one-sided event: under class $y{=}0$, misclassification requires $h_k \geq t$, i.e., $h_k - \mu_h^{(0)}(k) \geq \Delta(k)/2$; under class $y{=}1$, it requires $h_k \leq t$, i.e., $\mu_h^{(1)}(k) - h_k \geq \Delta(k)/2$. 
    Under Assumption~\ref{as:feature}, $h_k \mid Y{=}y$ is sub-Gaussian with variance  proxy $\sigma_h^2$, so the one-sided sub-Gaussian tail bound gives
    \begin{equation*}
      P\!\bigl(h_k - \mu_h^{(y)}(k) \geq \Delta(k)/2 \mid Y{=}y\bigr)
      \;\leq\;
      \exp\!\left(-\frac{\Delta(k)^2}{8\,\sigma_h^2}\right).
    \end{equation*}
    The same bound applies to the other class by symmetry. Averaging over the two classes preserves the bound. 
    The minimum prefix length follows by inverting the inequality $\exp(-\Delta(k)^2/(8\sigma_h^2)) \leq \varepsilon$.
\end{proof}

\begin{proof}[Proof of Proposition~\ref{prop:invariance}]
    AUROC is a rank statistic: it equals the probability that a randomly drawn positive instance receives a higher score than a randomly drawn negative instance, and therefore depends only on the rank ordering of classifier scores across samples.

    For classifiers whose decisions depend only on univariate feature ranks, each decision step compares the value of a single feature against a threshold. 
    Since rank preservation guarantees that the relative ordering of each feature's values is identical under both encoders, every such comparison yields the same outcome, and the final score ranking is preserved. 
    AUROC invariance follows immediately.

    For linear classifiers, the score is $s = \beta_0 + \sum_j \beta_j F_j(\phi(\mathcal{T}))$.
    Rank preservation of individual features does not in general imply rank preservation of a linear combination, since independent monotone transformations can alter relative feature scales. 
    However, if the within-class covariance structure of the feature vector is shared across encoders up to a global scalar $c > 0$ (i.e., $\Sigma_{\phi_2} = c\,\Sigma_{\phi_1}$), then the optimal coefficient vectors under the two encoders are proportional, and the resulting scores are related by a positive affine transformation. 
    Since AUROC is invariant to monotone score transformations, it is identical under both encoders.
\end{proof}

\end{document}